\documentclass[final,1p,times,authoryear]{elsarticle}
\usepackage{amsmath}
\usepackage{amsthm}
\usepackage{amssymb}
\usepackage{amsfonts}

\usepackage{xcolor,tikz,algorithm,mathtools}
\usetikzlibrary{matrix,calc,shapes,snakes,arrows,backgrounds,calc,fit,decorations.pathreplacing,decorations.markings,patterns,positioning,shapes.geometric}
\newdimen\colL
\newdimen\colR

\newcommand{\Z}{\mathbb{Z}}

\renewcommand{\Pr}[1]{{\mathop{\text{Pr}}\left[#1\right]}}

\newcommand{\Zn}{\Z_n}
\newcommand{\Znstar}{\Z_n^*}
\newcommand{\Zpstar}{\Z_p^*}

\def\multiset#1#2{\left(\!\left({#1\atopwithdelims..#2}\right)\!\right)}

\newtheorem{theorem}{Theorem}
\newtheorem{lemma}[theorem]{Lemma}
\newtheorem{corollary}[theorem]{Corollary}

\newtheorem{proposition}[theorem]{Proposition}
\newtheorem{definition}[theorem]{Definition}

\newtheorem{example}[theorem]{Example}

\begin{document}

\begin{frontmatter}

\title{Baby-Step Giant-Step Algorithms for the Symmetric Group}

\author{Eric Bach}
\address{University of Wisconsin-Madison}
\ead{bach@cs.wisc.edu}
\ead[url]{http://pages.cs.wisc.edu/~bach/}

\author{Bryce Sandlund}
\address{University of Wisconsin-Madison}
\ead{sandlund@cs.wisc.edu}
\ead[url]{http://pages.cs.wisc.edu/~sandlund/}

\begin{abstract}
We study discrete logarithms in the setting of group actions. Suppose that $G$
is a group that acts on a set $S$. When $r,s \in S$, a solution $g \in G$ to
$r^g = s$ can be thought of as a kind of logarithm. In this paper, we study
the case where $G = S_n$, and develop analogs to the Shanks baby-step /
giant-step procedure for ordinary discrete logarithms. Specifically, we
compute two sets $A, B \subseteq S_n$ such that every permutation of $S_n$ can be
written as a product $ab$ of elements $a \in A$ and $b \in B$. 
Our deterministic procedure is optimal up to constant factors, in the sense that $A$ and $B$ can
be computed in optimal asymptotic complexity, and $|A|$ and $|B|$ are a small constant from 
$\sqrt{n!}$ in size. We also analyze randomized ``collision'' algorithms for the same problem.
\end{abstract}

\end{frontmatter}

\section{Introduction}
Collision algorithms have been used to obtain polynomial 
(typically square root) speedups since the advent of computer science. 
Indeed, there are even collision ``algorithms'' in the world of
analog measurement \citep{AK}. Most collision algorithms
exploit time-space tradeoffs, arriving at a quicker algorithm by 
storing part of the search space in memory and utilizing an efficient
lookup scheme. One of the 
most famous of these collision-style methods is Shanks's 
baby-step giant-step procedure for the discrete 
logarithm problem \citep{SHANKS}.

Traditionally,
the discrete logarithm problem is the problem of finding an integer $k$ 
such that $b^k = g$, where $b$ and $g$ are elements of a finite cyclic 
group of order $n$ and $b$ is a generator (has order $n$). 
Then there is exactly one $k \in \{1,\ldots,n\}$ such that $b^k = g$.
Shanks's baby-step giant-step algorithm then 
writes $k = im + j$ with $m = \lceil{\sqrt{n}}\rceil$ 
and $0 \leq i, j \leq m$ and looks for a collision in the equation:
$$
g(b^{-m})^i = b^j.
$$
By precomputing values of $b^j$ (or $b^{-mi}$) and storing them in a hash table, 
a collision can be found in $O(\sqrt{n})$ time and $O(\sqrt{n})$ space, recovering the solution $k$.

Various extensions of the baby-step giant-step algorithm have been developed, 
mostly focusing on discrete logarithm problems in groups that are important
to cryptography.  For the classic problem, \cite{POL} contributed two elegant
methods that also exploit collision, but use very little space.  (They
have yet to be rigorously analyzed, in their original form.)
For more information about these algorithms, and more efficient methods that
apply to specific groups of an arithmetic nature, we refer to
surveys by \cite{MCC} and \cite{TES}.  

Ideas similar to the baby-step giant-step algorithm have been used
on 0-1 integer programming problems.  (This seems to be folklore.)
Suppose we want to solve $Ax=b$, where $x$ is a 0-1 column vector.
If we let $x_1$ and $x_2$ be half-length column vectors, and split
$A$ down the central column into $A_1$ and $A_2$, we can use collision
to solve $A_1 x_1 = b - A_2 x_2$.  (Here, we exploit not a group
structure, but rather, the associative law for matrix-vector multiplication.)
For a recent application, see \cite{ETH}.

In this paper, we focus on a different discrete log generalization
that can be stated as follows.

\begin{definition}
Suppose that $G$ is a finite group that acts on a set $S$.
Denote by $r^g$ the action of $g \in G$ on $r \in S$. Then, given elements
$r, s \in S$, the group action discrete logarithm problem is the problem of finding
a $g$ such that $r^g = s$.
\end{definition}


The first step beyond brute force search for this problem is 
to design an analog to the Shanks method. 
We will find appropriate splitting sets $A, B \subseteq G$ so that
for any $s$ in the orbit $r^G = \{r^g : g \in G\}$, we have
$r^{ab} = s$ for some $a \in A$ and $b \in B$. A match in the two sets
$$
\{r^a : a \in A\} \hbox{ and } \{s^{b^{-1}} : b \in B\}
$$
recovers the solution $g = ab$.

In this work, we treat two situations.  The first is one of maximum
generality: we know almost nothing about the structure of $G$, and can
only work with it by applying it to elements of $S$.  The second
is maximally specific: $G$ is the symmetric group $S_n$.
In neither case do we assume any particular knowledge about the orbit of $r$ or $s$.

For general groups (the first case), we analyze randomized methods 
that achieve square root speedups when compared to
the naive approach of exhaustive search.
For the important second case, we develop deterministic algorithms 
that utilize the structure of $S_n$.  These algorithms have
close to square root complexity.

To motivate our model and concentration on $S_n$, we give several
applications that fit into our framework.

\section{Applications}

We first show how group actions
lead to an unconventional algorithm for the graph isomorphism
problem (GI).
Let $S$ be the set of adjacency matrices for graphs
on $n$ vertices.  The symmetric group $S_n$ acts on $S$, via
$M^g = P_g M P_g^{-1}$.  ($P_g$ is just the permutation matrix
for $g$.)  In this case, the group action discrete logarithm problem
is exactly graph isomorphism: given adjacency matrices $M$ and $N$,
find $g \in S_n$ to make $M^g = N$, or determine no such $g$ exists.  Using our
results, we arrive at a deterministic graph isomorphism procedure
with run time about $\sqrt{n!}$.

Although there are much faster algorithms than this, they
are either conceptually involved \citep{BKL, B16}, cannot 
guarantee efficient performance
in all cases \citep{CFI}, or both \citep{NAUTY}.
The baby-step giant-step algorithm, on the other hand, 
gives an immediate proof that exhaustive search through permutations is 
not the best method for graph isomorphism.

There are a variety of other GI-related problems that also fit in the discrete log
symmetric group action framework. In particular, hypergraph isomorphism
 and equivalence of permutation groups via conjugation can both be
formulated as symmetric group actions. Furthermore, the latter problem has no
known moderately exponential ($\exp (n^{1-c})$ for $c > 0$) algorithm \citep{B16}.



We further note that in cryptography, the solution of iterated block ciphers \citep{MH} is
closely related to a group-action discrete logarithm problem in a symmetric group.
Our approach may also be useful in computational Galois theory,
specifically, in computing the splitting field of a polynomial \citep{ORY}.

Because our approach is orbit-oblivious, our framework is very general.
In some problems, however, algorithms aware of orbit restrictions may
benefit significantly by reducing the number of $g \in G$ to consider. This is true,
for example, in the graph isomorphism problem. An orbit-sensitive GI algorithm can utilize
the fact that a vertex can only be mapped to another vertex of the same degree, whereas
our approach will test every permutation.
On the other hand, there seem to be problems where one cannot exploit such orbit restrictions.

In our conference paper, we developed such an example using homomorphic encryption.
The idea was that if we take a group action discrete logarithm problem and convert the objects
$r$, $s$, and the action function $r^g$, to ciphertext $\mathcal{E}(r)$, $\mathcal{E}(s)$,
and an action function on ciphertext $\mathcal{E}(r)^g$, then no orbit restrictions can be determined
on the encrypted objects. Unfortunately, there is a slight issue with this example, because encryption and functions on
ciphertext are allowed to have different outputs for the same input; in fact, the security of such
cryptosystems require it. If this is the case, we can't compare ciphertexts for equality or use efficient
lookup schemes on ciphertext. The basic idea of the example still works, however, and we can
get around the issues by not explicitly using a homomorphic encryption system.

Instead, imagine that you are given two ordered arrays of $n$ indistinguishable objects with unique identities and a hash function
that maps each possible ordered array of $n$ objects uniquely to a positive integer. 
The problem is to compute the permutation that moves the objects of the first array into the same order
as the second, if such a permutation exists, using only the hash function on permutations of the first and second
arrays.
The symmetric group acts by
permuting arrays of indistinguishable objects, 
and we may store the integer given by the hash function for each
permuted array to determine a collision. In this case, this is all we can do, since each permutation of the
objects is indistinguishable and the associated integers say nothing about the particular
permutation represented. 
Our results show that we may solve this problem deterministically in time
and space $O(n \sqrt{n!})$.

Although the above problem is only theoretical, it is likely a realization exists with
the correct notion of indistinguishability and hash function that respects this notion.
Solving such group action discrete logarithm
problems by means of collision reduces to efficient construction of the sets $A$ and $B$.
In the remainder of this paper, we shift our focus to the construction of these sets, rather than
the mechanics of the algorithm itself. But before we do this,
we give a hardness result that shows the group action discrete logarithm problem is as hard
as the traditional discrete logarithm.

\section{Hardness}
In this section, we reduce the discrete logarithm problem to the group
action discrete logarithm problem.

Let $C_n$ be the cyclic group of order $n$, with generator $a$.  
The traditional discrete logarithm problem is the following.  Given
$b \in C_n$, find an $x \in \Zn$ for which $a^x = b$.  (Note that
$C_n$ is written multiplicatively.) This is not a group action 
discrete logarithm problem, because the $x$-th power map may not be 1-1:
consider $x=0$.

Every group, however, is acted on by its group of automorphisms.
In this way, we get an action of the group $\Znstar$ on $C_n$.
We now show how to reduce the traditional discrete logarithm
problem to computation of a group action discrete logarithm.

A randomized reduction is easiest to describe.  Given $a,b$
as above, choose a random integer $r$ with $0 \le r < n$, and
set $b' = b a^r$.  With probability $\Omega(1/(\log\log n))$,
the (unknown) discrete log $x$ of $a$ will satisfy
$x+r \in \Znstar$, thereby making
the randomly shifted element $ba^r$ a generator of
$C_n$.  A solution to the group action discrete log problem 
$a^y = b'$ then gives us $x$, as $y-r$.

This reduction can be made deterministic.  Recall that
Jacobsthal's function $J(n)$ denotes the maximum distance
between consecutive units of $\Zn$.  (Distance is reckoned along the cycle,
so that, for example, $\pm 1$ are 2 apart mod 7, giving $J(7) = 2$.)
\cite{HI} showed that $J(n) = O(\log n)^2$, although sharper
bounds have been conjectured, for example by \cite{RV}.
(No explicit bound that is polynomial in $\log n$ seems to be known,
but \cite{CW} come close: their estimate is
$(\log n)^{O(\log\log\log n)}$.) Therefore, we can replace the
random guess $r$ by a search through $r=1,2,\ldots$\ .
For one of these choices, $ba^r$ will generate $C_n$.  We use,
at most, $O(\log n)^2$ values of $r$.

For many applications in cryptography, $n$ is chosen to be
prime, and this case is of interest to us as well.  Here, there are
only two ``types'' for the value $x$.  As an element of $\Zn$,
$x$ is either 0, which we can test by comparing $b$ to 1, or relatively
prime to $n$.  So, the traditional discrete log problem in $C_n$ with $n$
prime reduces immediately to a group action discrete log problem.

Using these reductions, we can draw conclusions about the hardness
of the group action discrete log problem. \cite{VS} showed
that any generic probabilistic algorithm to solve the 
discrete log problem in $C_n$ must do $\Omega(\sqrt p)$ operations,
where $p$ is the largest prime factor of $n$.  
Informally, by a generic algorithm, we mean one that interacts with
the group only by doing group operations and equality tests.  (See
Shoup's paper for a precise definition.)  Earlier, \cite{VN}
had proved a similar result for a large class of deterministic algorithms.
Now, let us imagine that
we have a generic algorithm $Y$ that can solve the discrete log
problem $a^x = b$ in $C_n$, but only when $a$ and $b$ are generators.
(If the condition isn't fulfilled, the algorithm fails.)
In effect, $Y$ solves a group action discrete logarithm problem for
the case $G= \Znstar$ and $S = C_n$, under a certain promise about
the inputs.  Using our reductions, we can extend $Y$ to get a
generic algorithm $Y'$ that solves traditional discrete log problems,
with polynomial (in $\log n$) overhead.  The running time of $Y$
must therefore obey Shoup's lower bound: it must use $\Omega(\sqrt p)$
group operations.

Since generic algorithms can be randomized, a consequence of this
is that the black-box algorithm based off of section 4 is best possible,
since it applies to $S = C_p$ and $G = \Zpstar$.

\section{A Randomized Approach}
In this section our approach is general enough to work over any group $G$ 
where random elements can be generated efficiently. For the special case when $G = S_n$, 
this is possible through random shuffling procedures. 
For arbitrary permutation groups, this is also possible with an 
additive poly($n$) overhead, where $n$ is the degree of the 
permutation group \citep{SERESS}. The algorithmic approach to random sampling
 of an arbitrary (permutation) group is discussed more thoroughly in \cite{SERESS} and \cite{GC}.

In this setting, an obvious idea is to simply pick $k$ random elements of $G$ for the set $A$ 
and $k$ random elements of $G$ for the set $B$. 
Then, the probability a particular $g \in G$ is present in $AB$ will depend on the value of $k$. 
For ease of notation, let $m = |G|$. We have:

\begin{proposition}
Suppose we pick $k$ random elements of $G$ \textbf{without} replacement for the set $A$ and likewise for $B$. Then the probability a particular $g \in G$ is present in $AB$ satisfies:
\[
\Pr{g \in AB} \ge 1 - e^{-k^2/m}.
\]
\end{proposition}

\begin{proof}
Observe that $g \not\in AB$ precisely when each $b \in B$ 
avoids the set $\{a^{-1} g : a \in A \}$. The probability
of this event is
\[
\left( 1 - \frac k {m} \right)
\left( 1 - \frac k {m-1} \right)
\cdots
\left( 1 - \frac k {m-k+1} \right),
\]
which is at most $( 1 - k/m)^k$.
Rewriting the exponent and
utilizing that $(1-x/n)^n \le e^{-x}$ gives us
$$
\Pr{ g \notin AB } 
\le \left( \left(1-\frac{k}{m}\right)^{m} \right)^{k/m}
\le e^{-k^2/m},
$$
from which the claim follows.
\end{proof}

By setting $k = \Theta(\sqrt{m})$, we can make the probability that
$g$ is present in $AB$ constant.  Our analysis, however,
assumed sampling without replacement. If we simply sample with replacement
and redraw when a duplicate is found, it is not hard to see that as long
as we are sampling $o(m)$ elements, the number of extra draws is $O(1)$
in expectation.


Note that with this approach,
there will always be a non-zero probability that some group elements 
are missing in $AB$, which will lead to one-sided error in our algorithm.
Namely, if no $g$ is found where $r^g = s$, the randomized procedure only gives
probabilistic evidence that no $g$ exists. Furthermore, checking for 
missing elements of $G$ in $AB$ takes $O(|G|)$ time. While this would only 
need to be done once and work for any set $G$ acts on, 
it is prohibitively expensive. 

This leads us to ask for a deterministic algorithm to construct
the sets $A$ and $B$. This question inherently asks about structure of the
group $G$ that can be exploited, similarly to the original Shanks method for
the traditional discrete logarithm. Therefore, we will focus on the special
case when $G = S_n$.

\section{Background on Groups}

Our deterministic algorithm will rely on some elementary group theory. 
For this, we state a few necessary results and definitions.

All groups in this paper will be finite.  If $K$ is a subgroup
of $G$, we write $K \le G$, and $K < G$ if the containment is 
proper. We do not assume that $K$ is normal in $G$.

When $K \le G$, its left cosets are the $|G|/|K|$ sets
$gK$ with $g \in G$.  Right cosets are defined similarly.



Note that the set of left (or right) cosets of $K$ in $G$ forms
a partition of $G$.
Thus if we have a subgroup $K$ of $G$, we can take 
$B = K$ and $A$ to be a set of elements of 
$G$ such that each left coset of $K$ in $G$ is represented in $A$.
Then every element of $G$
will be present in $AB$.

In group theory, a minimal perfect set $A$ of this kind is called
a transversal.

\begin{definition}[Transversal]
A left (right)
transversal $T$ of a subgroup $K$ of $G$ is a set of 
elements of $G$ such that each left (right) coset of $K$ in $G$ 
has exactly one representative in $T$. Thus, $T$ is a minimal set of 
coset representatives of $K$ in $G$.
\end{definition}

To make this definition clear, we give the following:

\begin{example}
Let $G = \Z_{n^2}$ and let $K$ be the unique subgroup 
of $G$ with $n$ elements.  Then $|G : K| = n$.
If $G = \{0, \ldots, n^2-1\}$,
$K = \{0, n, 2n, \ldots, (n-1)n\}$. 
One transversal of $K$ in $G$ is 
$T = \{0, 1, 2, \ldots, n-1\}$.
\end{example}

For this example, $K$ and $T$ are exactly the sets of giant steps 
and baby steps that the Shanks algorithm would use.  However, transversals
are not unique; for example, we could have taken $T$ to be any
complete set of representatives modulo $n$.

Combinatorially, a subgroup $B$ and its left transversal $A$ form a perfect
splitting set for $G$, in the sense that every $g \in G$ is uniquely
of the form $ab$, for $a \in A$ and $b \in B$.  Perfect splitting sets
need not be subgroups, as we could always replace $A$ by $Ax$ and
$B$ by $x^{-1} B$, choosing $x \in G$ at will. 

Ideally, these splitting sets have cardinality exactly $\sqrt{|G|}$.
However, $n!$ is never a square for $n>1$, so such perfect splitting sets 
cannot exist for $G = S_n$. Therefore, we will either have to
tolerate duplicated products (as we did in the last section),
or look for set sizes close to, but not exactly matching, $\sqrt{n!}$.

We now give some concepts that provide
a ``data structure'' for working with permutation groups:

\begin{definition}[Base]
Let $G$ act on $\Omega$.  A base $B$ for $G$ is an ordered subset of $\Omega$ 
(i.e. a list) with the following property:
the only element of $G$ that stabilizes everything in $B$ 
is the identity.
\end{definition}

For our purposes, $G = S_n$, which stabilizes no element of 
$\Omega = \{1,\ldots,n\}$.
So, we will use $B := [1, 2, \ldots, n]$; that is, $B = \Omega$ with the natural ordering
of the integers. 
We could choose to not include any single integer from $B$, 
since the action of a permutation on the missing integer can be inferred 
via its action on the other elements; however, it will be easier to 
describe the transversal algorithm with a more complete base, and 
no loss of efficiency will be incurred. A base provides 
a convenient form to represent elements of $G$:

\begin{definition}[Base Image]
If $g \in G$ and
$B = [\beta_1, \beta_2, \ldots, \beta_k]$ is a base for $G$, 
then $B^g := [\beta_1^g, \ldots, \beta_k^g]$ 
is called the base image of $g$ (relative to $B$). 
\end{definition}

Recall that $\beta_i^g$ means the result of applying the 
group element $g$ to the object $\beta_i \in \Omega$.

With the base $B := [1, 2, \ldots, n]$, the base image gives the 
typical vector notation of a permutation. The base image
$B^g$ uniquely determines the element $g \in G$.




\section{Bidirectional Collisions}

In recent work, \cite{DR} described a general framework 
to create deterministic collision algorithms for isomorphism problems.
In his dissertation, Rosenbaum
does not initially give a way to apply this to permutation search, instead using a simpler
problem to demonstrate the technique
\citep[p. 196]{DR}. We are not aware that the application of bidirectional collision detection
to permutation search is known, so we describe and analyze
what we believe is the most natural application in this section.

Let $X$ and $Y$ be two objects that we wish to test for isomorphism.
In bidirectional collision detection, potential isomorphisms are represented by root-to-leaf
paths in a tree. The individual labelling of nodes at a particular level is not known, and is in general
different between $X$ and $Y$.
We choose one set of paths that reach
all nodes halfway down the tree, extending to leaves arbitrarily,
and apply these to $X$.  We choose one path to an arbitrary
node halfway down the tree, extend it in all possible ways, and then apply each
of these to $Y$. The shared path represents the isomorphism between $X$ and $Y$.

\begin{figure}[!htb]
\label{searchtree}
\begin{center}
\begin{tikzpicture}[scale=.9]
\tikzstyle{vertex}=[circle,draw,minimum size=15pt,inner sep=0pt]
\tikzstyle{edge} = [draw]
\tikzstyle{weight} = [font=\small]

\node[vertex] (n1) at (-5.5, 0) {};
\node[vertex] (n2) at (-4.5, 0) {};
\node[vertex] (n3) at (-3.5, 0) {};
\node[vertex] (n4) at (-2.5, 0) {};
\node[vertex] (n5) at (-1.5, 0) {};
\node[vertex] (n6) at (-0.5, 0) {};
\node[vertex] (n7) at (0.5, 0) {};
\node[vertex] (n8) at (1.5, 0) {};
\node[vertex] (n9) at (2.5, 0) {};
\node[vertex] (n10) at (3.5, 0) {};
\node[vertex] (n11) at (4.5, 0) {};
\node[vertex] (n12) at (5.5, 0) {};

\node[vertex] (n13) at (-4.5, 1.5) {};
\node[vertex] (n14) at (-1.5, 1.5) {};
\node[vertex] (n15) at (1.5, 1.5) {};
\node[vertex] (n16) at (4.5, 1.5) {};

\node[vertex] (n17) at (0, 3) {};

\path[edge,blue,thick] (n13) -- (n17) {};
\path[edge,blue,thick] (n15) -- (n17) {};
\path[edge,blue,thick] (n16) -- (n17) {};

\path[edge,blue,thick] (n1) -- (n13) {};
\path[edge,blue,thick] (n9) -- (n15) {};
\path[edge,blue,thick] (n10) -- (n16) {};

\path[edge,green,thick] (n5) -- (n14) {};
\path[edge,green,thick] (n14) -- (n17) {};

\path[edge,red,thick] (n4) -- (n14) {};
\path[edge,red,thick] (n6) -- (n14) {};

\path[edge] (n2) -- (n13) {};
\path[edge] (n3) -- (n13) {};
\path[edge] (n7) -- (n15) {};
\path[edge] (n8) -- (n15) {};
\path[edge] (n11) -- (n16) {};
\path[edge] (n12) -- (n16) {};

\end{tikzpicture}
\caption{An example search tree for bidirectional collision detection.
The set of paths applied to $X$ are in blue,
to $Y$ are in red, and the shared path is in green.}
\end{center}
\end{figure}
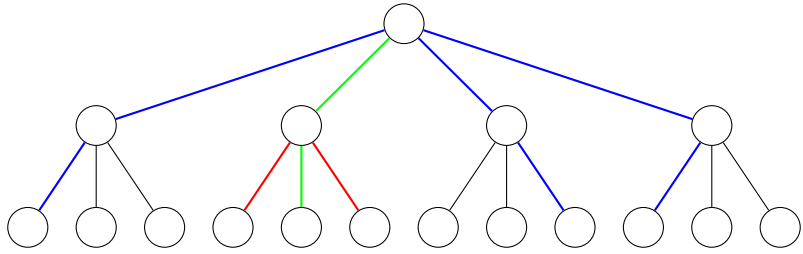

If we denote the set of transformations applied to $X$ as
$C$ and to $Y$ as $D$, then, borrowing the group action notation,
this framework finds paths $c \in C$ and $d \in D$ so that
$$
X^c = Y^d
$$
whenever there exists some $g$ such that
$$
X^g = Y.
$$
Thus, this approach writes $g = cd^{-1}$, for some $c \in C$ and $d \in D$.
This means the sets $C$ and $D^{-1}$ form a factorization of the search space in
the same way as our sets $A$ and $B$. Note that by $D^{-1}$, we mean the set of
all inverses of the elements in $D$.

To split the search space for permutations, the first level of our search tree will
correspond to where we send the integer $1$, the second the integer $2$, and so on.
We choose a $k$ with 
$1 \le k \le n$ as our halfway point.
The set $C$ will consist of $(n)_k$ ($n$ pick $k$) permutations. 
Each permutation sends $1, 2, \ldots, k$ to all possible $k$-tuples in $\{1, 2, \ldots, n\}$.
We can choose arbitrarily where to send $k+1, \ldots, n$.

To make $D$, we choose an
image tuple for the first $k$ elements arbitrarily (not moving them at all 
will do) and then extend with all $(n - k)!$ possible suffixes.

Since $|C| = (n)_k$ and $|D| = (n-k)!$, the counting functions for
these two sets are not interchangable.  However, we can try to balance
the values of $(n)_k$ and $(n-k)!$.  Since $(n)_k \cdot (n-k)! = n!$,
finding a $k$ such that $(n)_k \approx (n-k)!$ will,
at least approximately, minimize $(n)_k + (n-k)!$, the cost 
of the collision algorithm.  We will show that the ``right'' value
of $k$ is roughly, but not exactly, $n/2$.  Surprisingly, 
the performance of the algorithm is very sensitive to $k$.

To proceed, we must obtain asymptotics for the solution $x$
in $x! = \sqrt{n!}$. We will use $x!$ as an abbreviation for $\Gamma(x+1)$ and
$\log$ as the natural logarithm.

\begin{proposition}
\label{xval}
For integer $n \ge 1$, let $x$ be the positive real solution to $x! = \sqrt{n!}$.
Then
\[
x = \frac n  2 \left(1 + O\left(\frac 1 {\log n}\right) \right).
\]
\end{proposition}
\begin{proof}
Since log factorial is convex \citep[p. 13]{AAR}, and the central binomial 
coefficient is a positive integer, we have
$$
{\left({\frac n 2}\right)!}^2
\le  \left \lfloor \frac n 2 \right \rfloor !  \cdot \left \lceil \frac {n}{2} \right \rceil ! 
\le n! .
$$
Therefore, $n/2 \le x \le n$.
From the enveloping property of Stirling's series \citep[pp. 252-253]{WW}, we 
deduce that
$$
    \left({\frac z e}\right)^z \le z! \le z^z 
$$
holds for $z \ge 1$.  This implies
$$
x \left( \log x - 1 \right) \le \log x! \le \frac n 2 \log n,
$$
so 
$$
x \le \frac n 2 \frac{\log n}{\log x - 1}
  \le \frac n 2 \frac{\log n}{\log n - \log 2 - 1}
  =  \frac n 2 \left( 1 + O\left( \frac 1 {\log n} \right) \right).
$$
Since $x \ge n/2$ as well, the claimed result follows.
\end{proof}

We can get more precise results than the above by taking an accurate form
of Stirling's formula, such as $\log z! = z \log z - z + O(\log z)$,
and then ``bootstrapping'' \citep{GK}.  For example,
$$
x = \frac n 2 \left( 1 + \frac {\log 2} {\log n}
               + O\left( (\log n)^{-2} \right) \right).
$$
This is already pretty good: if $n=10$,
$x \approx 6.509$, whereas 
$(n/2) ( 1 + \log 2 / \log n ) \approx 6.505$.

We now return to determining the best value of $k$ for bidirectional collision detection
applied to permutation search, and analyze the performance of the technique with this choice.

\begin{proposition}
Bidirectional collision detection applied to permutation search finds $A,B \subseteq S_n$ 
such that $AB = S_n$ and 
$\max(|A|, |B|)$ $=$ $\Theta( n^{1/2} \sqrt{n!})$.
\end{proposition}

\begin{proof}
Let $x$ be the positive real solution to $x! = \sqrt{n!}$.
Since $x!$ increases for $x \ge 1$, there is an integer $m$ such that
\[
(m-1)! \leq x! = \sqrt{n!} \leq m!.
\]
Choose $k$ so that $n-k$ 
is the closest integer (either $m$ or $m-1$) to $x$.  This will cause
one of our sets to be larger than its ``ideal'' value $\sqrt{n!}$, 
and we must estimate this disparity.  Let
$$
n-k = x + \alpha, \qquad \hbox{ with } |\alpha| \le 1/2.
$$
Recall that $(x+\alpha)! \sim x! x^\alpha$ as $x \rightarrow \infty$,
in the sense that the limiting ratio is 1.  Using this,  and
our asymptotic expression for $x$ from Proposition \ref{xval}, we have
$$
(n-k)! = (x+\alpha)! \sim x! x^{\alpha} 
\sim \sqrt{n!} \left( \frac n 2 \right)^\alpha.
$$
Similarly,
$$
(n)_k = \frac {n!}{(n-k)!} \sim \frac{n!} {\sqrt{n!} (n/2)^\alpha}
      = {\sqrt{n!} (n/2)^{-\alpha}}.
$$
The bound on $|\alpha|$ gives
$$
\max(|A|,|B|) =
\max( (n)_k, (n-k)! ) = \Theta\left( n^{1/2} \sqrt{n!} \right).
$$
\end{proof}
   
Probably, the factor $n^{1/2}$ is best possible.  Examination
of numerical data shows that $\alpha$ varies irregularly within
the interval $(-1/2,1/2)$, and we see no reason why this behavior
should not continue.  In particular, there is likely to be an
infinite sequence of $n$'s on which $\alpha$'s limiting value
is $1/2$.

It is also interesting to compare our procedure to one
that splits the permutation vectors exactly in half, 
as Rosenbaum initially recommends \cite[p. 193]{DR}. 
This splitting
amounts to taking $k=n/2$, and as a consequence of
Stirling's formula,
$$
(n)_{n/2} = \Theta(n^{-1/4} 2^{n/2} \sqrt{n!})
$$
when $n$ is even.  Therefore, the ``overhead'' for exact splitting 
is exponential in $n$. It is not hard to see this intuitively,
by considering a decision tree for generating permutations. Each
of the ``large'' branching factors $n-i$, $0 \le i < n/2$, is
roughly twice as large as its ``small'' counterpart $n/2 - i$.
By choosing the best splitting fraction for each $n$,
we have reduced this overhead to a small power of $n$.

Finally, we note that this approach indeed employs a subgroup
and a corresponding transversal of that subgroup. The set $D$
fixes the first $k$ elements and permutes the remaining $n-k$ 
elements amongst themselves in all possible ways.  So $D$
is isomorphic to $S_{n-k}$. Because of this, $D^{-1} = D$. 
Then, since $CD = S_n$ and $|C||D| = n!$, $C$ is a perfect set of
left coset representatives of $D$ in $S_n$, i.e., a left transversal of $S_{n-k}$
in $S_n$.

\section{A Better Subgroup}

In the conference version of this paper, we improved upon bidirectional
collision detection by finding a subgroup of $S_n$ of size
$\Theta(n^{\pm 1/4} \sqrt{n!})$, with the plus or minus sign depending on the parity of $n$.
With its corresponding transversal, this leads to sets $A$ and $B$ of size 
$\max(|A|,|B|) = \Theta(n^{1/4} \sqrt{n!})$.
In this version, we find a subgroup of size $\Theta(\sqrt{n!})$, leading to
sets $A$ and $B$ of size $\max(|A|,|B|) = \Theta(\sqrt{n!})$, which is optimal up to constant factors.

Our conference paper relied on a particular subgroup that was very close to $\sqrt{n!}$ in size.
Our new approach will be more similar to the previous section. In particular, instead of
 choosing a specific
subgroup for each $n$, we will choose a subgroup from a set of options with sizes roughly evenly
(geometrically) distributed from $1$ to $n!$.
To improve upon the $O(\sqrt{n})$ size gap given by bidirectional
collision detection, we will need to pick our subgroup from a larger set of options.

One very simple idea is to take the symmetric group $S_k$ for $0 \leq k \leq n$ and add an
$\ell$-cycle on the remaining $n-k$ integers into the generating set, where $\ell$ may
range from $0$
to $n-k$. When $\ell = 0$ or $1$, we take the subgroup to be $S_k$; further, when $k = 0$ or $1$,
we take the subgroup to be the powers of the $\ell$-cycle.
Without loss of generality, we assume either $k$ or $\ell$ is greater than $0$.

We will show that for large enough $n$, there exists
a choice of $k$ and $\ell$ that produces a subgroup of this structure with size
very close to $\sqrt{n!}$.

\begin{lemma}
\label{subgroup}
Denote by $\sigma$ the $\ell$-cycle $(k+1\ k+2\ldots\ k+\ell)$ and by 
$H$ the subgroup of $S_n$ generated by $S_k$ and $\sigma$.
Here, we require $0 \le k \le n$ and $0 \le \ell \le n-k$. WLOG, we assume $k + \ell > 0$.
Then $|H| = k!$ when $\ell = 0$ or $1$ and $|H| = \ell k!$ when $\ell \geq 2$.
Explicitly, the subgroup $H$ is generated by
\[H = \left\{\begin{array}{cl}
	\langle (1\ 2), (1\ 2\ldots\ k), (k+1\ k+2\ldots\ k+\ell) \rangle &: 
	          \mathrm{if}\ \ell \geq 2\ \mathrm{and}\ k \ge 2\\
	\langle (1\ 2), (1\ 2\ldots\ k) \rangle &: 
	          \mathrm{if}\ \ell = 0\ \mathrm{or}\ 1\ \mathrm{and}\ k \ge 2\\
	\langle (2\ 3\ldots \ell+1) \rangle &:
	          \mathrm{if}\ k = 1.\\
	\langle (1\ 2\ldots \ell) \rangle &:
	          \mathrm{if}\ k = 0.\\
\end{array} \right. \]

\end{lemma}
\begin{proof}
Recall that $S_k$ is generated by $(1\ 2)$ and $(1\ 2\ldots\ k)$.

When $k = 0$ or $1$, $H$ is the powers of the $\ell$-cycle $\sigma$, so
$H = \langle (1\ 2\ldots \ell) \rangle$ if $k = 0$ and $H = \langle (2\ 3\ldots \ell+1) \rangle$ if $k = 1$.

When $\ell = 0$ or $1$ and $k \geq 2$, $S_k$ is our entire subgroup, therefore 
$H = \langle (1\ 2), (1\ 2\ldots\ k) \rangle$, and (even if $k = 0$ or $1$) $|H| = k!$.

Now assume $\ell \geq 2$ and $k \ge 2$.
The $\ell$-cycle $\sigma = (k+1\ k+2\ldots\ k+\ell)$ is disjoint from $S_k$, so we may compose any element
of $S_k$ with a power of $\sigma$ to produce a new element of $H$. Thus, in this case,
$H = \langle (1\ 2), (1\ 2\ldots\ k), (k+1\ k+2\ldots\ k+\ell) \rangle$.
Since there are $\ell$
distinct powers of $\sigma$, when $\ell \geq 2$, $|H| = \ell k!$.
\end{proof}

Now suppose we are looking for a subgroup of $S_n$ of size $m$. Assume $m < n!$. If $m = n!$, we may use
$S_n$ itself. Let $x$ be the largest integer such that
\[
x! \leq m < (x+1)!.
\]
We will show the following.

\begin{theorem}
\label{HDistribution}
We may find a subgroup $H$ of $S_n$ of size within a factor of
\[
\max \left( \sqrt{2}, \sqrt{\frac{x + 1}{n - x}} \right)
\]
of $m$.
\end{theorem}

\begin{proof}
By Lemma \ref{subgroup}, we may find a subgroup $H$ for any $x+\ell \leq n$.
Thus, we may find a subgroup of size $\ell x!$ for $1 \leq \ell \leq n-x$ and $(x+1)!$.
Consider the sizes of these subgroups and the ratios of sizes between consecutive subgroups:
\begin{center}
\includegraphics[scale=.43]{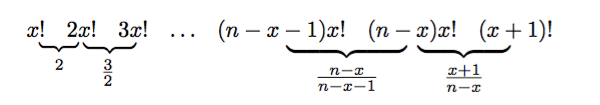}
\end{center}
The extremal ratios are $2$ and $(x+1)/(n-x)$. Let
\[
\alpha = \max \left( \sqrt{2}, \sqrt{\frac{x + 1}{n - x}} \right).
\]
Then $m$ will reside within a factor of $\alpha$ from the size of one of the above subgroups;
that is, we may find a subgroup $H$ with
\[
\frac{m}{\alpha} \leq |H| \leq \alpha m.
\]
\end{proof}

Using Theorem \ref{HDistribution}, we may prove the following Corollary 
for our specific choice of $m = \sqrt{n!}$.

\begin{corollary}
\label{optimal}
We may find a subgroup $H$ of $S_n$ of size $\frac{\sqrt{n!}}{\sqrt{2}} \leq |H| \leq \sqrt{2n!}$ for $n \geq 7$; that is, we may find a subgroup $H$ of size $\Theta(\sqrt{n!})$ 
in $S_n$.
\end{corollary}
\begin{proof}
For $0 < x < n$, the derivative with respect to $x$ of $(x+1)/(n-x)$ is $(n+1)/(n-x)^2$,
 so this function is increasing on $(0, n)$. Solving for $x$ in $\sqrt{2} = \sqrt{(x+1)/(n-x)}$ 
 yields $x = (2n-1)/3$. Therefore, whenever $0 < x \leq (2n-1)/3$,
 \[
 \max \left( \sqrt{2}, \sqrt{\frac{x + 1}{n - x}} \right) = \sqrt{2}.
 \]
 Now let $y$ be the positive real solution to $y! = \sqrt{n!}$. As proven in Proposition \ref{xval},
 the solution can be expressed
 asymptotically by
 \[
 y = \frac n 2 \left( 1 + O \left( \frac 1 {\log n} \right) \right).
 \]
 If we let $x = \lfloor y \rfloor$, as in Theorem \ref{HDistribution}, then for all sufficiently large
 $n$, $x$ will be less than or equal to $(2n-1)/3$. Examination of numerical data shows that
 for $n \geq 7$, $\lfloor y \rfloor \leq \lfloor (2n-1)/3 \rfloor$.
\end{proof}

The above corollary says that there exists a subgroup $H$ within a constant factor of
 $\sqrt{n!}$ for
sufficiently large $n$. To find it, we may simply try all $k$ and $\ell$ with $1 \leq k+\ell \leq n$,
$\ell \geq 0$, and
choose the values that make $\ell k!$ closest to $\sqrt{n!}$. Charging $O(1)$ time to multiply integers,
this can be done in $O(n^2)$ time, which will be much less than the cost of enumerating $H$ or
its transversal.

Once these values are determined,
we may run a closure algorithm using the generators from Lemma \ref{subgroup} to enumerate the
elements of $H$ in $O(n|H|)$ time (the $O(n)$ factor is for representing the permutations
themselves). There is an additional factor of $2$ or $3$ overhead with this approach because
each permutation needs to be composed with the generating set before it can be determined all
elements have been enumerated.
This overhead can be avoided by
instead using the structure of $H$ to construct the elements explicitly. We briefly
give this algorithm.

\begin{algorithm}
\caption{Enumeration of the elements of $H$}
\label{Henumeration}
\begin{enumerate}
  \item Enumerate all elements of $S_k$. If $k = 0$, enumerate the identity.
  \item For each permutation from 1, 
  output its composition with all powers of $\sigma = (k+1\ k+2\ldots\ k+\ell)$. If $\ell = 0$,
  simply output the permutation from 1.
\end{enumerate}
\end{algorithm}

Before analyzing the algorithm, we make a few statements regarding our model of computation. 
For the purposes of this paper, we will charge $O(1)$ space for each integer and $O(1)$ time for 
accessing elements of an array. We will not charge space for the output of an algorithm.
In this model, we will represent each permutation using $O(n)$ space 
and may compose two permutations in $O(n)$ time.

\begin{lemma}
Algorithm \ref{Henumeration} correctly enumerates the elements of $H$ in time $O(n |H|)$ 
and space $O(n)$.
\end{lemma}
\begin{proof}
For correctness, Lemma \ref{subgroup} shows the set returned is $H$.

To argue resource bounds, note that we may iterate through every permutation of $S_k$ in 
$O(n k!)$ time
and $O(n)$ space \citep[p. 112]{HOLT}.
For every permutation generated, we may then iterate through the
powers of $\sigma$ in total time $O(n \ell k!)$ and again space $O(n)$. 
Composing the pairs of permutations from 1 and 2 again takes total time $O(n \ell k!)$
and no additional space, thus in total the algorithm takes time $O(n \ell k!) = O(n |H|)$
 and space $O(n)$.
\end{proof}

\section{Constructing Transversals}
\label{transversals}

We now consider finding a transversal of $H$ in $S_n$.  Although
finding transversals can be complicated and require backtrack search
through the parent group $G$ \citep{HOLT}, we can take advantage of
having $G=S_n$.

We first give the following definition, as in \cite{HOLT}.

\begin{definition}
Let $B = [\beta_1,\ldots,\beta_k]$ be a base.
Define a partial ordering $\prec$ on elements of $\Omega$ by 
taking $\beta_i \prec \beta_{i+1}$ for all $1 \leq i < k$,
and $\beta_i \prec \alpha$, for every $\alpha \in \Omega$
not present in $B$.  We extend this to base images by
saying that for $g, h \in G$, $B^g \prec B^h$ if 
$g$ precedes $h$ in the lexicographic ordering on the base 
vectors.
\end{definition}

For our purposes, this is the natural ordering 
of the integers in $B$. Furthermore, this defines lexicographical 
ordering of permutations for base images 
$B^g$ and $B^h$, $g, h \in G$.

Our transversal construction will exploit the following lemma.

\begin{lemma}
\label{CosetOrdering}
Let $K < G \leq S_n$ and let 
$B = [\beta_1, \ldots, \beta_k]$ be a base of $G$. 
Then $g \in G$ is the $\prec$-least element of its coset $gK$ if 
and only if $\beta_j^g$ is the $\prec$-least element of its 
orbit in $K_{\beta_1^g, \ldots, \beta_{j-1}^g}$ for $1 \leq j \leq k$. 
\end{lemma}

In this lemma, $K_{\beta_1^g, \ldots, \beta_{j-1}^g}$ denotes
the subgroup of $K$ consisting of the elements that fix each 
of the elements listed as subscripts. (When $j=1$, this subgroup
is just $K$.) 

Although the result has been known since the work of Charles Sims,
we present a proof in the interest of being self-contained.
This lemma can be found (without proof) in \cite[p. 115]{HOLT}.

\begin{proof}
Suppose that $g$ satisfies the property given in Lemma \ref{CosetOrdering}. We must show $g$ is $\prec$-least in $gK$. Write:
\[
g = [\alpha_1, \alpha_2, \ldots, \alpha_k].
\]
Now suppose there exists some $h \in gK$ such that $h \prec g$. Write:
\[
h = [\gamma_1, \gamma_2, \ldots, \gamma_k].
\]
Since $h \in gK$, we can write $h = gk$ for some $k \in K$. Therefore we can think of $h$ as applying some element $k$ to $g$. Now, since $h \prec g$, there must be a first index $j$ such that $\gamma_j \prec \alpha_j$; so $\gamma_i = \alpha_i$ for all $1 \leq i < j$. Then $k$ must stabilize $\alpha_1, \ldots, \alpha_{j-1}$. Since $\gamma_j \prec \alpha_j$, $k$ must map $\alpha_j$ to $\gamma_j$, therefore $\gamma_j$ and $\alpha_j$ are in the same orbit in $K_{\alpha_1, \ldots, \alpha_{j-1}}$. But by assumption, $\alpha_j$ is the $\prec$-least such element in its orbit in $K_{\alpha_1, \ldots, \alpha_{j-1}}$. Therefore the element $h$ cannot exist and so $g$ is minimal in $gK$.

In the other direction, suppose $g$ does not satisfy the property in Lemma \ref{CosetOrdering}. Let $j$ be the first index such that $\alpha_j$ is not $\prec$-least in its orbit in $K_{\alpha_1, \ldots, \alpha_{j-1}}$. Then there must be some $k \in K$ that stabilizes $\alpha_1, \ldots, \alpha_{j-1}$ and maps $\alpha_j$ to some $\eta$ such that $\eta \prec \alpha_j$. Then $gk \prec g$. 
\end{proof}

We can apply Lemma \ref{CosetOrdering} to find base images that satisfy the property required to be $\prec$-least elements in their respective cosets. However, these base images might not necessarily correspond to elements of $G$ if $G$ is an arbitrary permutation group. Here we take advantage of the fact $G = S_n$. Every base image that satisfies Lemma \ref{CosetOrdering} corresponds to some permutation on $\{1, 2, \ldots, n\}$, which is necessarily an element of $S_n$. Thus, we may compute a transversal
of $H$ efficiently if know the orbit structure after stabilizing in $H$.

\begin{lemma}
\label{HStabilizers}
The subgroup $H$ is initially composed of at most three categories of orbits, which further split
after stabilizing the smallest integer in the following way:
\begin{enumerate}[(a)]
  \item Integers $1, 2, \ldots k$ are in an orbit such that after repeatedly stabilizing
  the smallest integer, the remaining integers continue to stay in a single orbit.
  \item Integers $k+1, k+2, \ldots, k+\ell$ are in an orbit such that after stabilizing $k+1$,
  integers $k+2, k+3, \ldots, k+\ell$ are all in their own orbits.
  \item Integers $k+\ell+1, k+\ell+2, \ldots, n$ are in their own orbits.
\end{enumerate}
If $k = 0$, (a) disappears. If $\ell = 0$, (b) disappears. If $k + \ell = n$, (c) disappears.
\end{lemma}

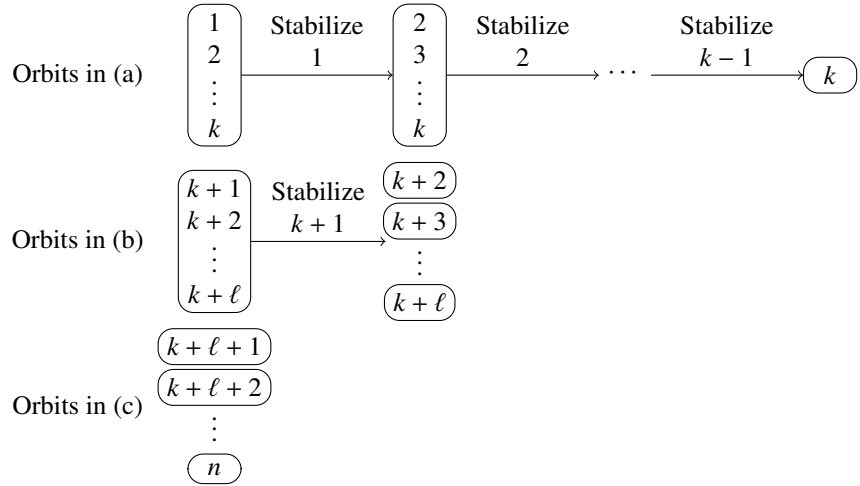
\begin{figure}[!htb]
\begin{center}
\begin{tikzpicture}[scale=.9]


\node at (-2, 1.75) {Orbits in (a)};
\node at (-2, -.7) {Orbits in (b)};
\node at (-2, -3.15) {Orbits in (c)};

\node[draw, rectangle,align=center, rounded corners=.2cm, minimum width=.7cm] (a1) at (0, 1.75) {$1$\\ $2$\\ $\vdots$\\ $k$};
\node[draw, rectangle,align=center, rounded corners=.2cm, minimum width=.7cm] (a2) at (3, 1.75) {$2$\\ $3$\\ $\vdots$\\ $k$};
\path (a1) edge[->] node[above,black,align=center] {Stabilize\\ $1$} (a2);
\node (a3) at (6, 1.75) {$\cdots$};
\path (a2) edge[->] node[above,black,align=center] {Stabilize\\ $2$} (a3);
\node[draw, rectangle,align=center, rounded corners=.2cm, minimum width=.7cm] (a4) at (9, 1.75) {$k$};
\path (a3) edge[->] node[above,black,align=center] {Stabilize\\ $k-1$} (a4);

\node[draw, rectangle,align=center, rounded corners=.2cm] (b1) at (0, -.7) {$k+1$\\ $k+2$\\ $\vdots$\\ $k+\ell$};
\node[draw, rectangle,align=center, rounded corners=.2cm] (b2a) at (3, .2) {$k+2$};
\node[draw, rectangle,align=center, rounded corners=.2cm] (b2a) at (3, -.4) {$k+3$};
\node (b2c) at (3, -.9) {$\vdots$};
\node[draw, rectangle,align=center, rounded corners=.2cm] (b2a) at (3, -1.6) {$k+\ell$};
\path (b1) edge[->] node[above,black,align=center] {Stabilize\\ $k+1$} (2.5, -.7);

\node[draw, rectangle,align=center, rounded corners=.2cm] (c1) at (0, -2.25) {$k+\ell+1$};
\node[draw, rectangle,align=center, rounded corners=.2cm] (c2) at (0, -2.85) {$k+\ell+2$};
\node (c3) at (0, -3.35) {$\vdots$};
\node[draw, rectangle,align=center, rounded corners=.2cm,minimum width=.7cm] (c4) at (0, -4.05) {$n$};


\end{tikzpicture}
\caption{The orbit structure in $H$}
\end{center}
\end{figure}

\begin{proof}
The subgroup $H$ is generated by $S_k$ and $(k+1\ k+2\ldots\ k+\ell)$.
Therefore, the integers $1, 2, \ldots k$ are in an orbit, $k+1, k+2, \ldots k+\ell$ are in an orbit,
and $k+\ell+1, k+\ell+2, \ldots, n$ are in their own orbits. If $k = 0$, there are no integers in this
first category. If $\ell = 0$, there are no integers in this second category. If $k + \ell = n$, there
are no integers in this third category.

Now consider the first orbit in category (a).
When we stabilize any integer, we are left with a subgroup isomorphic to $S_{k-1}$.
All other integers remain in the same orbit, and the structure will repeat itself if we continue
to stabilize integers.

Now consider orbit (b). This orbit comes from powers of $\sigma = (k+1\ k+2\ldots\ k+\ell)$. Every power
of $\sigma$ moves every integer in $\{k+1, k+2, \ldots, k+\ell\}$. Therefore, when
we stabilize $k+1$, no powers of $\sigma$ will be in this stabilizer. Thus,
after stabilizing $k+1$, integers $k+2, k+3, \ldots, k+\ell$ will be in their own orbits.

Once an integer is in its own orbit, it will remain in its own orbit after further stabilizations.
\end{proof}

Using Lemma \ref{CosetOrdering} and Lemma \ref{HStabilizers}, the algorithm to generate a transversal
of $H$ in $S_n$ is relatively straightforward. We wish to backtrack through the orbits in all
possible ways such that each base image generated through the procedure has the property of
being $\prec$-least for its respective coset.
For our subgroup, this means the integers $\{1, 2, \ldots, k\}$ must appear in this
order and $k+1$ must appear before any of $\{k+2, k+3, \ldots, k+\ell\}$.
Any backtracking procedure that enumerates all base images with this property suffices to compute
a transversal of $H$ in $S_n$. We give one such backtracking procedure below.

\begin{algorithm}[H]
\caption{Transversal of $H$ in $S_n$}
\label{Htransversal}
\begin{enumerate}
  \item Let $A = [1, 2, \ldots, k]$. If $k=0$, $A = []$.
  \item Enumerate all permutations on $\{k+1, k+2, \ldots, k+\ell\}$ that start with integer $k+1$.
           If $\ell = 0$, enumerate an ``empty'' permutation, $[]$.
  \item For each permutation from 2,
           enumerate all permutations on $\{k+\ell+1, k+\ell+2, \ldots, n\}$.
           If $k+\ell = n$, enumerate an ``empty'' permutation, $[]$.
  \item For each pair from steps 2 and 3, combine them in all possible ways with each other
           and $A$ such that the integers from steps 1, 2, and 3 remain in the same relative order.
           For each combination, output the resulting permutation.
\end{enumerate}
\end{algorithm}

\begin{lemma}
\label{Htransversal-lemma}
Algorithm \ref{Htransversal} correctly enumerates a transversal of $H$ in $S_n$ in time
$O(n |S_n : H|)$ and space $O(n)$.
\end{lemma}
\begin{proof}
Every permutation generated in the above procedure respects that integers $\{1, 2, \ldots, k\}$
appear in this same order and that $k+1$ appears before any of $\{k+2, k+3, \ldots, k+\ell\}$.
Furthermore, every permutation that respects this order is present in the returned set. To double check,
we may count the number of permutations generated from the above procedure.

Steps 2 and 3 enumerate $(\ell - 1)!$ and $(n - (k+\ell))!$ permutations, respectively.
For each pair from these steps, we merge lists in all possible ways while still respecting relative
order within each list. Let
\[
\multiset{x}{y} = {x + y - 1 \choose y}
\]
denote the number of $y$-multisets taken from an $x$-set.
Then, using a stars and bars argument, the number of ways to merge lists while still
respecting relative order is
\[
\multiset{k+1}{\ell} \multiset{k+\ell+1}{n-(k+\ell)} = {k + \ell \choose \ell} {n \choose n-(k+\ell)}.
\]
Putting it together, we have
\begin{align*}
|S_n : H| &= (\ell-1)! (n-(k+\ell))! {k + \ell \choose \ell} {n \choose n-(k+\ell)}\\
&= (\ell-1)! (n-(k+\ell))! \frac{(k+\ell)!}{\ell!k!} \frac{n!}{(n-(k+\ell))!(k+\ell)!}\\
&= \frac{n!}{\ell k!},
\end{align*}
which is exactly what we expect.

To prove resource bounds, observe that for steps 2 and 3, we may use a procedure to iterate
through permutations for a total time cost of $O(n(\ell-1)!(n-(k+\ell)))!$ and space cost of $O(n)$.
To perform step 4, we may use a backtracking procedure to output all possible combinations
in space $O(n)$ and with an additional factor of $n$ time cost on the number of elements produced,
for a total time cost of $O(n |S_n : H|)$ and total space cost of $O(n)$.
\end{proof}

Finally, we note that the property exploited in computing the transversal was that every base image found directly via stabilizers and orbits is necessarily an element of the parent group $G$ if $G = S_n$. This observation can be combined with black-box orbit, stabilizer, and base changing methods as discussed in \cite{HOLT} to compute a transversal of an arbitrary permutation group $K < S_n$ efficiently. We will not discuss the details here, nor give exact asymptotic guarantees. For more information, consult the transversal algorithms discussed in Holt's book.

\section{The Main Result}
We have the following theorem:

\begin{theorem}
\label{AB}
We can compute sets $A$ and $B$ such that $AB = S_n$ with 
$\max(|A|, |B|) = \Theta(\sqrt{n!})$. 
The computation can be done deterministically in time $O(n \sqrt{n!})$ and space $O(n)$.
\end{theorem}

\begin{proof}
Corollary \ref{optimal} states that for $n \geq 7$, we may find a subgroup $H$ within a
factor of $\sqrt{2}$ from $\sqrt{n!}$. Using algorithm \ref{Henumeration}, it can
be enumerated deterministically in time $O(n|H|)$ and space $O(n)$. Using Algorithm \ref{Htransversal},
its transversal can be found and enumerated deterministically in time $O(n|S_n : H|)$
and again space $O(n)$.
Therefore, we can find sets of permutations $A$ and $B$ 
deterministically in time $O(n \sqrt{n!})$ and space $O(n)$
such that every permutation of $S_n$ can 
be represented as a product $ab$, $a \in A$ and 
$b \in B$.
\end{proof}

This lemma implies that the group action discrete logarithm problem in the symmetric group can be 
solved deterministically in about $O(n \sqrt{n!})$ time and $O(n \sqrt{n!})$ space, where the specifics 
depend on the ability
to hash or compare elements of $S$ generated for the collision procedure.
We note that while randomization may be used in the analysis of such hashing functions, 
the algorithm itself will always produce correct results.

Furthermore, regarding the space cost of solving the group action discrete logarithm problem,
we require that one of the two sets
$$
\{r^a : a \in A\} \hbox{ or } \{s^{b^{-1}} : b \in B\}
$$
be stored in memory. By storing the smaller set, we may take advantage of a time-space tradeoff.
Since the bounds in Corollary \ref{optimal}
hold for any choice of $m \le \sqrt{n!}$, we may choose $H$ within a $\sqrt{2}$ factor
of any size $m \le \sqrt{n!}$. Then
we may solve the group action discrete logarithm problem
 in the symmetric group deterministically
in space $O(m)$ and time $O(\frac{n!}{m})$ for any choice of $m \le \sqrt{n!}$.

The computation of sets $A$ and $B$, in comparison, will also take time $O(n \sqrt{n!})$ by the 
randomized approach and time $O(n^{1.5} \sqrt{n!})$ by bidirectional collision detection.
In the former case, it is informative to know for what error bounds it becomes more efficient to
use our deterministic method.

Our analysis of the randomized approach, applied to $S_n$,
shows that by picking $k$ random permutations for $A$ and $B$,
$\Pr{ g \notin AB } \le e^{-k^2/n!}$. As $n!$ grows large, this inequality becomes increasingly tighter.
Taking it as a baseline for the probability of missing a particular permutation in $AB$, if we set
$k = \sqrt{2n!}$, this probability is
\[
\Pr{ g \notin AB } \approx e^{-2} \approx .1353.
\]
Thus, it is more efficient to use our deterministic method if we want more than about 86.47\%
accuracy; that is, when the algorithm determines no such $g$ exists, this response is correct about
86.47\% of the time.

On another practical note, examination of numerical data shows the best $H$ for 
$n = 1$ to $30$ is, at worst, off from $\sqrt{n!}$
by a factor of $\sqrt{2}$ when $n = 2$ (as it must) and approximately
$1.2883$ when $n = 28$. On average, from $n = 1$ to $30$, the best choice of $H$
has size only off by a factor of approximately $1.1201$ of $\sqrt{n!}$.
Combining these facts with the simplicity of Algorithm \ref{Henumeration} and Algorithm
\ref{Htransversal}, we consider our approach a viable practical alternative to the randomized one.

\section{Acknowledgements}

This research was supported in part by NSF: CCF-1420750.
We'd like to thank Derek Holt, Gene Cooperman, and L\'{a}szl\'{o} Babai for correspondence on computing
transversals in permutation groups.  We additionally thank the stackexchange 
community for their input on subgroup choices, general help with computational group theory,
and innovative latex solutions. Finally, we thank the anonymous reviewers
for their constructive feedback during the review process.

\bibliographystyle{elsarticle-harv}


\end{document}